\documentclass[10pt,conference]{IEEEtran}\IEEEoverridecommandlockouts
\usepackage{epsfig,graphicx,subfigure,psfrag,amsmath,cases}
\usepackage{latexsym,amssymb,epsfig,subfigure,algorithm}
\usepackage{algorithmic}
\usepackage{color}
\usepackage{url}
\usepackage{scrtime}
\usepackage{bm}
\usepackage{framed}
\author{\authorblockN{Shiyang Leng\authorrefmark{4}, Derrick Wing Kwan Ng\authorrefmark{1}, Nikola Zlatanov\authorrefmark{2}, and Robert Schober\authorrefmark{3}\thanks{ This work was supported in part by the AvH Professorship Program of the Alexander von Humboldt Foundation.}}
The Pennsylvania State University, USA\authorrefmark{4}\\
The University of New South Wales, Australia\authorrefmark{1}\\
Monash University, Australia\authorrefmark{2}\\
Friedrich-Alexander-University Erlangen-N\"urnberg (FAU), Germany\authorrefmark{3}\\}

\title{Multi-Objective Beamforming for Energy-Efficient SWIPT Systems}

\date{\thistime,\,\today}

\newtheorem{Thm}{Theorem}
\newtheorem{Lem}{Lemma}

\newtheorem{Prob}{Problem}
\newtheorem{T-Prob}{Transformed Problem}
\newtheorem{proposition}{Proposition}
\DeclareMathOperator{\Tr}{\mathrm{Tr}}
\DeclareMathOperator{\zero}{\mathbf{0}}
\DeclareMathOperator{\Rank}{\mathrm{Rank}}

\DeclareMathOperator{\vect}{\mathrm{vec}}
\DeclareMathOperator{\maxo}{\mathrm{maximize}}
\DeclareMathOperator{\mino}{\mathrm{minimize}}

\newcommand{\est}[1]{{E}\!\left\{#1\right\}}

\newcommand{\textoverline}[1]{$\overline{\mbox{#1}}$}

\newcommand{\abs}[1]{\lvert#1\rvert}

\newcommand{\norm}[1]{\lVert#1\rVert}

\begin{document}
\maketitle

\begin{abstract}
In this paper, we study the resource allocation algorithm design for energy-efficient simultaneous wireless information and power transfer (SWIPT) systems. The considered system comprises a transmitter, an information receiver, and multiple energy harvesting receivers equipped with multiple antennas. We propose a multi-objective optimization framework to study the trade-off between the maximization of the energy efficiency of information transmission and the maximization of wireless power transfer efficiency. The proposed problem formulation takes into account the per  antenna circuit power
consumption of the transmitter and the imperfect channel state information of the energy harvesting receivers. The adopted  non-convex multi-objective optimization problem is  transformed into an equivalent rank-constrained semidefinite program (SDP) and  optimally solved by  SDP relaxation.  Numerical results unveil an interesting trade-off between the considered
conflicting system design objectives and reveal the benefits of multiple transmit antennas
for improving system energy efficiency.
\end{abstract}
\renewcommand{\baselinestretch}{1}
\section{Introduction}
In recent years, the development of wireless communication networks worldwide has triggered an exponential growth in the number of wireless devices and sensors for applications such as e-health and environmental monitoring. The related tremendous increase in the number of  transmitter(s) and receiver(s) has also led to a huge demand for energy and a better energy management.  Hence,    energy efficient system designs, which adopt energy efficiency (bit-per-Joule) as the performance
metric, have been recently proposed \cite{Magazine:green1}-\nocite{JR:energy_efficient_1,JR:energy_efficient_2}\cite{JR:hybrid_BS}. In \cite{JR:energy_efficient_1}, energy-efficient power
allocation schemes were proposed for cognitive radio systems. In \cite{JR:energy_efficient_2}, energy-efficient link adaptation
was investigated for the maximization of
energy efficiency in frequency-selective channels. In \cite{JR:hybrid_BS}, the authors proposed a resource allocation algorithm design for energy-efficient communication in multicarrier communication systems with hybrid
energy harvesting base stations.  Although energy-efficient resource allocation algorithm designs for traditional communication networks  have been studied in the literature, mobile receivers are often powered by batteries with limited
energy storage  which remain the system performance bottlenecks in perpetuating
the lifetime of wireless networks.

Energy harvesting (EH) based communication system design is a viable solution for prolonging the lifetime of energy-limited devices. Conventional natural sources, such as wind, solar, and biomass, have been exploited as energy sources for fixed-location, outdoor transmitters.  However, these  natural energy sources are often location and weather dependent and may not be suitable for mobile receivers. On the other hand, wireless power transfer (WPT) via electromagnetic waves in radio frequency (RF) enables a comparatively controllable energy harvesting for mobile receivers. In fact, recent progress in the development of RF-EH circuitries has made RF-EH practical for  low-power consumption devices \cite{Krikidis2014}\nocite{JR:Rui_magazine}--\cite{JR:Xiaoming_magazine}, e.g. wireless sensors. Besides, WPT enables the dual use of the information carrier for simultaneous wireless information and power transfer (SWIPT) \cite{CN:WIPT_fundamental}\nocite{JR:WIP_receiver}--\cite{JR:non_linear_SWIPT}. Different from traditional wireless communication systems, where data rate and energy efficiency are the most fundamental system performance metrics, in SWIPT systems, the wireless energy transfer efficiency is an equally important QoS metric. Thus, the design of resource allocation algorithms should take into account the emerging need for energy transfer efficiency. In \cite{JR:MIMO_WIPT}, the authors studied the fundamental rate-energy trade-off region for optimal information beamforming. In \cite{JR:WIPT_fullpaper_OFDMA}, power allocation, user scheduling, and subcarrier allocation were jointly designed to enable an energy-efficient multicarrier SWIPT system.  In \cite{JR:EE_SWIPT_Massive_MIMO}, the authors proposed the use of large scale multiple-antenna systems for improving energy efficiency of  SWIPT.  Although energy-efficient data communication design and energy-efficient WPT have already been studied individually, the trade-off between these two system design paradigms is still unclear for SWIPT systems. In particular, these two design goals may conflict with each other but both are desirable to system designer. However,  the single-objective resource allocation algorithms proposed in \cite{JR:energy_efficient_1}--\cite{JR:hybrid_BS}, \cite{CN:WIPT_fundamental}--\cite{JR:MIMO_WIPT,JR:EE_SWIPT_Massive_MIMO}  may no longer be applicable in energy-efficient SWIPT networks.


In this paper, we address the above issues. To this end,
we formulate the resource allocation algorithm design as a multi-objective optimization problem which strikes a balance between  the maximization of energy efficiency of information transmission  and the maximization of WPT efficiency.
The resulting non-convex optimization problem is solved
optimally by semidefinite programming (SDP) relaxation. Simulation results illustrate the trade-off between the conflicting
system design objectives.
\section{System Model}
In this section, we first define the adopted notations and then present the  channel model for energy-efficient SWIPT networks.

\subsection{Notation}

$\mathbf{A}^H$, $\Tr(\mathbf{A})$, and $\Rank(\mathbf{A})$ represent the  Hermitian transpose, trace, and rank of  matrix $\mathbf{A}$; $\mathbf{A}\succeq \mathbf{0}$ indicates that $\mathbf{A}$ is a  positive semidefinite matrix; matrix $\mathbf{I}_{N}$
denotes an $N\times N$ identity matrix. $\vect(\mathbf{A})$ denotes the vectorization of matrix $\mathbf{A}$ by stacking its columns from left to right to form a column vector.
$\mathbf{A}\otimes \mathbf{B}$ denotes the Kronecker product of matrices $\mathbf{A}$ and $ \mathbf{B}$. $[\mathbf{B}]_{a:b,c:d}$ returns the $a$-th to the $b$-th rows and the $c$-th to the $d$-th columns block submatrix of $\mathbf{B}$.  $\mathbb{C}^{N\times M}$ denotes the space of $N\times M$ matrices with complex entries.
$\mathbb{H}^N$ represents the set of all $N$-by-$N$ complex Hermitian matrices.
The distribution of a circularly symmetric complex Gaussian (CSCG)
vector with mean vector $\mathbf{x}$ and covariance matrix
$\mathbf{\Sigma}$  is denoted by ${\cal
CN}(\mathbf{x},\mathbf{\Sigma})$, and $\sim$ means ``distributed
as".  $\cal E\{\cdot\}$ denotes statistical expectation.  $\norm{\cdot}$ and $\norm{\cdot}_F$ denote the Euclidean norm and the Frobenius norm of a vector/matrix, respectively.   $\mathrm{Re}(\cdot)$ extracts the real part of a complex-valued input.

\begin{figure}
\centering\vspace*{-5mm}
\includegraphics[width=3.5 in]{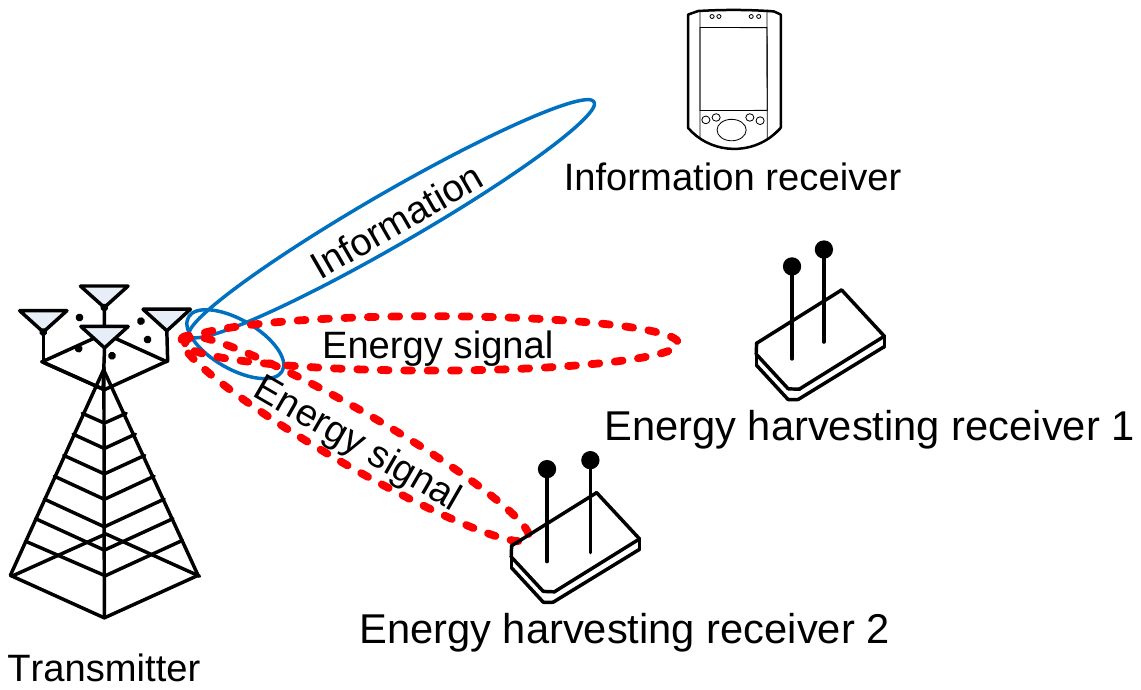}
\caption{A downlink communication system with an information receiver (IR) and $J=2$ energy harvesting receivers (ERs).}
\label{fig:system_model2}\vspace*{-5mm}
\end{figure}
\subsection{Channel Model}
We focus on a downlink  SWIPT system. The system consists of a transmitter, a single-antenna information receiver (IR), and multiple energy harvesting receivers (ERs). The transmitter is equipped with $N_\mathrm{T}$ antennas and each ER is equipped with $N_{\mathrm{R}}$ receiving antennas. We assume that the ERs are roaming  wireless terminals from other communication systems searching for additional power supply in the RF. The transmission is divided into time slots. In each time slot, the transmitter sends a precoded information signal and an energy signal simultaneously to facilitate information transmission to the IR and power transfer to the ERs, cf. Figure \ref{fig:system_model2}. The transmit signal is given by
\begin{eqnarray}
\mathbf{x}=\mathbf{w}_\mathrm{I}s+\mathbf{w}_\mathrm{E},
\end{eqnarray}
where $s\in \mathbb{C}$ is the information-bearing symbol with $\est{{\abs{s}}^2}=1$ and $\mathbf{w}_\mathrm{I}\in \mathbb{C}^{{N_\mathrm{T}}\times 1}$ is the corresponding information beamforming vector. $\mathbf{w}_\mathrm{E}\in \mathbb{C}^{{N_\mathrm{T}}\times 1}$ is the energy signal facilitating energy transfer to the ERs. The energy signal $\mathbf{w}_\mathrm{E}$ is  a deterministic pseudo-random sequence with zero mean and covariance matrix
$\mathbf{W}_\mathrm{E}$. Since $\mathbf{w}_\mathrm{E}$ is generated at the transmitter by a pseudo-random sequence generator with a predefined seed, the energy signal is known to the IR. Thus, the interference caused by the energy signal can be completely cancelled at the IR.

We assume a narrow-band slow fading channel between the transmitter and receivers. Then, the received signals at the IR and ER $j$ are expressed as
\begin{eqnarray}
y^\mathrm{IR}=\mathbf{h}^H(\mathbf{w}_\mathrm{I}s+\mathbf{w}_\mathrm{E})+n_\mathrm{I},\\
\mathbf{y}^\mathrm{ER}_j=\mathbf{G}^H_j(\mathbf{w}_\mathrm{I}s+\mathbf{w}_\mathrm{E})+\mathbf{n}^\mathrm{E}_j,
\end{eqnarray}
where $\mathbf{h}\in \mathbb{C}^{N_{\mathrm{T}}\times 1}$ is the channel vector between the transmitter and the IR, and $\mathbf{G}_j\in \mathbb{C}^{N_{\mathrm{T}}\times N_{\mathrm{R}}}$ is the channel matrix between the transmitter and ER $j$. Variables $\mathbf{h}$ and $\mathbf{G}_j$ capture the joint effect of multipath fading and path loss. $n_\mathrm{I}\in \mathbb{C}$ and $\mathbf{n}^\mathrm{E}_j\in \mathbb{C}^{ N_{\mathrm{R}}\times 1}$ are additive white Gaussian noises (AWGNs) at the IR and ER $j$, respectively, and are distributed as ${\cal CN}(0,\sigma_\mathrm{I}^2)$ and ${\cal CN}(\zero,\sigma_\mathrm{E}^2\mathbf{I}_{N_\mathrm{R}})$.

\section{Resource Allocation Algorithm Design}
In this section, we present the adopted performance metrics and the problem formulation.

\subsection{Achievable Rate, Harvested Energy, and Energy Efficiency}
 The achievable rate (bit-per-second) at the IR  is given by
\begin{eqnarray}\label{eqn:C}
R=B\log_2\Big(1+\frac{\mathbf{w}_{\mathrm{I}}^H\mathbf{H}\mathbf{w}_\mathrm{I}}{\sigma_\mathrm{I}^2}\Big),
\end{eqnarray}
where $B$ is the system bandwidth and $\mathbf{H}=\mathbf{h}\mathbf{h}^H$. We note that the interference caused by the energy signal, i.e., $\Tr(\mathbf{w}_{\mathrm{E}}^H\mathbf{H}\mathbf{w}_\mathrm{E})$, is removed from the IR via successive interference cancellation before the IR decodes the desired information  since the energy signal $\mathbf{w}_\mathrm{E}$ is known to the
receiver.   On the other hand, both the information signal and the energy signal can act as RF energy sources for the ER due to the broadcast nature of wireless channels. As a result, the total harvested energy\footnote{We note that a normalized energy unit, i.e., Joule-per-second, is adopted.  Therefore,
the terms ``power" and ``energy" are used interchangeably in this paper.} at ER $j$ is given by
\begin{eqnarray}\label{eqn:per_ER}
P^\mathrm{harv}_j&=&\eta_j\Tr\Big(\mathbf{G}_j^H\big(\mathbf{w}_\mathrm{I}\mathbf{w}^H_\mathrm{I}
+\mathbf{W}_{\mathrm{E}}\big)\mathbf{G}_j\Big),
\end{eqnarray}
and $\eta_j$ is the energy conversion efficiency of ER $j$ which is a constant with $0\leq\eta_j\leq1$ and models the energy loss of the process of converting the received RF energy to electrical energy for storage. We ignore the thermal noise at the receiving antenna for energy harvesting as it is relatively small compared to the received signal power.

 Energy efficiency is a fundamental system performance metric in modern communication networks. To design a resource allocation algorithm for
energy-efficient communication, the total power consumption
has to be included in the optimization objective function.
Thus, we model the power dissipation (Joule-per-second) in the system as
\begin{eqnarray}\label{eqn:Ptotal}
P_\mathrm{tot}&=&\frac{\norm{\mathbf{w}_\mathrm{I}}^2+\Tr(\mathbf{W}_\mathrm{E})}{{\xi}}+P_{\mathrm{B}},\\
\mbox{where}\,\, P_{\mathrm{B}}&=&N_\mathrm{T}P_\mathrm{ant}+P_\mathrm{c}.\label{eqn:PB}
\end{eqnarray}
\noindent $0< \xi\leq1$ is the constant power amplifier efficiency. The first term in (\ref{eqn:Ptotal}) is the total power consumption in the power amplifier. ${N_{\mathrm{T}}P_{\mathrm{ant}}}$ in \eqref{eqn:PB} accounts for the dynamic circuit power consumption which is proportional to the number of transmit antennas. $P_\mathrm{ant}$ denotes the power dissipation at each transmit antenna, including the dissipation in the transmit filter, mixer, frequency synthesizer, digital-to-analog converter (DAC), etc. $P_{\mathrm{c}}$ denotes the fixed circuit power consumption for baseband signal processing.

Therefore, the achievable rate energy efficiency (AR-EE) and the energy transfer energy efficiency (ET-EE) of the considered system are defined  as
\begin{eqnarray}
\Phi_\mathrm{IR}
&=&\frac{B\log_2(1+\frac{\mathbf{w}^H_\mathrm{I}\mathbf{H}\mathbf{w}_\mathrm{I}}{\sigma_\mathrm{I}^2})}{(\norm{\mathbf{w}_\mathrm{I}}^2
+\Tr(\mathbf{W}_\mathrm{E}))/{\xi}+N_\mathrm{T}P_\mathrm{ant}+P_\mathrm{c}}\mbox{ and}\\
\Phi_\mathrm{EH}
&=&\frac{\sum_{j}P^\mathrm{harv}_j}{(\norm{\mathbf{w}_\mathrm{I}}^2+\Tr(\mathbf{W}_\mathrm{E}))/
{\xi}+N_\mathrm{T}P_\mathrm{ant}+P_\mathrm{c}},\label{eq:ET-EE}
\end{eqnarray}
respectively, where $P^\mathrm{harv}_j$ is given in \eqref{eqn:per_ER}.

\subsection{Channel State Information (CSI)}
In this paper, we focus on a Time Division
Duplex (TDD)  communication system with slowly time-varying channels.  At the
beginning of each time slot,  handshaking is performed between the transmitter and the IR. As a result, the downlink CSI of the IR can be obtained by measuring the uplink training sequences embedded in the handshaking signals. Thus, we assume that the transmitter-to-IR fading gain, $\mathbf{h}$, can be estimated perfectly at the transmitter. On the other hand,   the ERs may not directly interact with the transmitter. Besides, the ERs may be silent for long periods of time. As a result,  the CSI of the ERs  can  be  obtained only occasionally at the transmitter when the ERs communicate with the transmitter. Hence, the CSI for the ERs may be outdated when the transmitter performs resource allocation.  We adopt a deterministic model \cite{JR:CSI-determinisitic-model,JR:Kwan_secure_imperfect} to characterize the impact of the CSI imperfection for resource allocation design. The  CSI of the link between the transmitter
and ER $j$ is modeled as
\begin{eqnarray}\label{eqn:outdated_CSI}
\mathbf{G}_j&=&\mathbf{\widehat G}_j + \Delta\mathbf{G}_j,\,   \forall j\in\{1,\ldots,J\}, \mbox{   and}\\
{\bm\Psi }_j&\triangleq& \Big\{\Delta\mathbf{G}_j\in \mathbb{C}^{N_{\mathrm{T}}\times N_{\mathrm{R}}}  :\norm{\Delta\mathbf{G}_j}_F^2 \le \varepsilon_j^2\Big\},\forall j,\label{eqn:outdated_CSI-set}
\end{eqnarray}
where $\mathbf{\widehat G}_j\in\mathbb{C}^{N_{\mathrm{T}}\times N_{\mathrm{R}}}$ is the  matrix CSI estimate of the channel of ER $j$ that is available at the  transmitter. $\Delta\mathbf{G}_j$ represents the unknown channel uncertainty and the continuous set ${\bm\Psi }_j$ in (\ref{eqn:outdated_CSI-set})  defines a continuous space spanned by all possible channel uncertainties. Constant $\varepsilon_j$ represents the maximum value of the norm of the CSI estimation error matrix  $ \Delta\mathbf{G}_j$ for
ER $j$.
\section{Problem Formulation and Solution}

\subsection{Problem Formulation}
In SWIPT systems, AR-EE maximization and ET-EE maximization are both desirable system design objectives. In this section, we first propose two problem formulations for single-objective system design for SWIPT. Each single-objective problem describes one important aspect of the system design. Then, we consider both system design objectives jointly via the multi-objective problem formulation.

The first system design objective is the maximization of AR-EE  without the consideration of energy harvesting. The corresponding optimization problem is formulated\footnote{We note that the considered problem formulation can be easily extended to the case with a minimum data rate requirement. Yet, a stringent data rate requirement does not facilitate the study the trade-off between different system objectives due to the resulting smaller feasible solution set.} as
\begin{framed}
\begin{Prob}AR-EE Maximization:\label{prob:WIPT_AR-EE}
\begin{eqnarray}
\underset{\mathbf{W}_{\mathrm{E}}\in\mathbb{H}^{N_\mathrm{T}},\mathbf{w}_{\mathrm{I}}}{\maxo}\,\, \hspace*{-2mm}&&\Phi_{\mathrm{IR}}\\
\hspace*{-2mm}\mathrm{subject\,\,to}\,\, &&\mathrm{C1}:\,\,\norm{\mathbf{w}_\mathrm{I}}^2+\Tr(\mathbf{W}_{\mathrm{E}})\leq P_{\mathrm{max}},\notag\\
\hspace*{-2mm}&&\mathrm{C2}:\,\,\mathbf{W}_{\mathrm{E}}\succeq \mathbf{0}\notag
\end{eqnarray}
\end{Prob}
\end{framed}\noindent
$P_{\mathrm{max}}$ in constraint C1 denotes the maximum transmit power budget. In addition, covariance matrix $\mathbf{W}_{\mathrm{E}}$ is a positive semidefinite Hermitian matrix as indicated by constraint C2.

The second system design objective is the maximization of the ET-EE. The corresponding problem formulation is given as
\begin{framed}
\begin{Prob}ET-EE Maximization:\label{prob:WIPT_ET-EE}
\begin{eqnarray}
\underset{\mathbf{W}_{\mathrm{E}}\in\mathbb{H}^{N_\mathrm{T}},\mathbf{w}_{\mathrm{I}}}{\maxo}\,\, && \underset{\Delta\mathbf{G}_j{\in\bm\Psi_j}}{\min}\,\,\Phi_{\mathrm{EH}}\\
\mathrm{subject\,\,to}\,\, &&\mathrm{C1,\,C2}.\notag
\end{eqnarray}
\end{Prob}
\end{framed}

For the sake of notational simplicity, we denote the objective functions in the above problems as $F_n$, $n=1,2$. In practice, these two system design objectives are both desirable from the system operator perspective. However, it is expected that there is a non-trivial trade-off between these objectives. In order to meet these conflicting system design objectives systematically and simultaneously, we adopt the
weighted Tchebycheff method for the multi-objective optimization \cite{JR:MOOP} which can provide the complete Pareto optimal set by varying predefined preference parameters. To this end, we incorporate the two individual system design objectives into a multi-objective optimization problem (MOOP), which is formulated as
\begin{framed}
\begin{Prob}Multi-Objective Optimization Problem:\label{prob:multiobj_WIPTsepuser}
\begin{eqnarray}
\underset{\mathbf{W}_{\mathrm{E}}\in\mathbb{H}^{N_\mathrm{T}},\mathbf{w}_{\mathrm{I}}}{\mino}\,\,&&\max_{i=1,2}\,\, \Big\{\omega_i(F_i^*-F_i)\Big\}\notag\\
\mathrm{subject\,\,to} &&\mathrm{C1,\,C2},
\end{eqnarray}
\end{Prob}\end{framed}
\noindent where $F_i^*$ is the optimal objective value with respect to Problem $i$. $\omega_i$ is a weight imposed on objective function $i$ subject to $0\leq\omega_i\leq1$ and $\sum_i\omega_i=1$, which indicates the  preference of the system designer for the $i$-th objective function over the others. In the extreme case, when $\omega_i=1$ and $\omega_n=0, \forall n\neq i $, Problem \ref{prob:multiobj_WIPTsepuser} is equivalent to single-objective optimization problem $i$.

\section{Optimization Solutions}
It can be observed that the objective functions of Problems \ref{prob:WIPT_AR-EE}--\ref{prob:multiobj_WIPTsepuser} are non-convex functions. In general, there is no well-known systematical approach for solving non-convex optimization problems. In order to obtain a tractable solution, we first transform the non-convex objective functions using the Charnes-Cooper transformation. Then, we use semidefinite programming relaxation (SDR) to obtain the resource allocation solution for the reformulated problem.

We first reformulate the aforementioned three optimization problems by defining a set of new optimization variables:
\begin{eqnarray}\label{eqn:newvariabledefine}
\mathbf{W}_\mathrm{I}&=&\mathbf{w}_\mathrm{I}\mathbf{w}_\mathrm{I}^H,\,\,
\theta=\frac{1}{P_\mathrm{tot}},\notag\\
\,\,\overline{\mathbf{W}}_\mathrm{I}&=&\theta\mathbf{W}_\mathrm{I},\,\,\mathrm{and}\,\,\overline{\mathbf{W}}_\mathrm{E}=\theta\mathbf{W}_\mathrm{E}.
\end{eqnarray}
Then, the original problems can be rewritten with respect to the new optimization variables $\{\overline{\mathbf{W}}_\mathrm{I},\overline{\mathbf{W}}_\mathrm{E}, \theta\}$. Problem 1 becomes
\begin{framed}
\begin{Prob}Transformed AR-EE Maximization Problem:\label{prob:WIPT_AR-EE_reform}
\begin{eqnarray}
\underset{\overline{\mathbf{W}}_\mathrm{I},\overline{\mathbf{W}}_\mathrm{E}\in\mathbb{H}^{N_\mathrm{T}},\theta}{\maxo}\,\, &&\hspace*{-5mm}\theta\log_2\Big(1+\frac{\Tr(\mathbf{H}\overline{\mathbf{W}}_\mathrm{I})}{\theta\sigma_\mathrm{I}^2}\Big)\\
\mathrm{subject\,\,to}\,\, &&\hspace*{-5mm}\overline{\mathrm{C1}}:\,\Tr(\overline{\mathbf{W}}_\mathrm{I}+\overline{\mathbf{W}}_\mathrm{E})\leq\theta P_{\mathrm{max}},\notag\\
&&\hspace*{-5mm}\overline{\mathrm{C2}}:\,\overline{\mathbf{W}}_\mathrm{I}\succeq \mathbf{0},\,\,\overline{\mathbf{W}}_\mathrm{E}\succeq \mathbf{0},\notag\\
&&\hspace*{-5mm}\overline{\mathrm{C3}}:\,\Rank(\overline{\mathbf{W}}_\mathrm{I})\leq1,\notag\\
&&\hspace*{-5mm}\overline{\mathrm{C4}}:\,\frac{\Tr(\overline{\mathbf{W}}_\mathrm{I}+
\overline{\mathbf{W}}_\mathrm{E})}{\xi}+\theta P_{\mathrm{B}}\leq1,\notag\\
&&\hspace*{-5mm}\overline{\mathrm{C5}}:\,\theta\ge0, \notag
\end{eqnarray}
\end{Prob}
\end{framed}\noindent
where $\overline{\mathbf{W}}_\mathrm{I}\succeq\zero, \overline{\mathbf{W}}_\mathrm{I}\in\mathbb{H}^{N_\mathrm{T}}$, and $\overline{\mathrm{C3}}$  are imposed to guarantee that $\overline{\mathbf{W}}_\mathrm{I}=\theta\mathbf{w}_\mathrm{I}\mathbf{w}_\mathrm{I}^H$. Constant $B$ is dropped from the objective function in Problem \ref{prob:WIPT_AR-EE_reform} since it is independent of the optimization variables. Similarly, Problem 2 becomes
\begin{framed}
\begin{Prob}Transformed ET-EE Maximization Problem:\label{prob:WIPT_ET-EE_reform}
\begin{eqnarray}
\underset{\overline{\mathbf{W}}_\mathrm{I},\overline{\mathbf{W}}_\mathrm{E}\in\mathbb{H}^{N_\mathrm{T}},\theta,\gamma_j}{\maxo}\,\, && \sum_{j=1}^J \gamma_j\\
\mathrm{subject\,\,to}\,\, &&\overline{\mathrm{C1}} - \overline{\mathrm{C5}}.\notag\\
&&\overline{\mathrm{C6}}:  \gamma_j\leq \underset{\Delta\mathbf{G}_j{\in\bm\Psi_j}}{\min}\,\, P^\mathrm{harv}_j,\forall j, \notag
\end{eqnarray}
\end{Prob}\end{framed}\noindent where $\gamma_j$ are auxiliary optimization variables.

Finally, Problem 3 can be written as
 \begin{framed}
\begin{Prob}Transformed MOOP:\label{prob:multiobj_WIPTsepuser2}
\begin{eqnarray}
\underset{\overline{\mathbf{W}}_\mathrm{I},\overline{\mathbf{W}}_\mathrm{E}\in\mathbb{H}^{N_\mathrm{T}},\theta,\tau}{\mino}&&\hspace*{-5mm}\tau \\
\mathrm{subject\,\,to}\,\,\,&&\hspace*{-5mm}\overline{\mathrm{C1}} - \overline{\mathrm{C6}},\notag\\
&&\hspace*{-5mm}\overline{\mathrm{C7}}:\,\omega_i(F_i^*-\hspace*{-0.5mm}F_i)\leq \tau,i\in\{4,5\},\notag
\end{eqnarray}
\end{Prob}\end{framed}
\noindent where $\tau$ is an auxiliary optimization variable.

\begin{proposition}\label{prop:equivalency}
The Problems \ref{prob:WIPT_AR-EE_reform}-\ref{prob:multiobj_WIPTsepuser2} are equivalent transformations of the original Problems \ref{prob:WIPT_AR-EE}-\ref{prob:multiobj_WIPTsepuser}, respectively.
\end{proposition}
\begin{proof}
The transformation is based on  Charnes-Cooper transformation. Due to the space limitation, we refer to \cite{JR:MOOP} for  proof for a similar problem.
\end{proof}

We note that  Problem \ref{prob:multiobj_WIPTsepuser2} is a generalization of Problems 4 and 5. If Problem \ref{prob:multiobj_WIPTsepuser2} can be solved optimally by an algorithm, then the algorithm can also be used to solve Problems 4 and 5. Thus, we focus on the method for solving\footnote{In studying the solution structure of Problem 6, we assume that the optimal objective values of Problems 4, 5 are given constants, i.e.,  $F^*_p,\forall p\in\{4,5\}$, are known. Once the structure of the  optimal resource allocation scheme  of Problem 6 is obtained, it can be exploited to obtain the optimal solution of  Problems 4, 5. } Problem \ref{prob:multiobj_WIPTsepuser2}. It is evident that Problem \ref{prob:multiobj_WIPTsepuser2} is non-convex due to the rank-one beamforming matrix constraint $\overline{\mathrm{C3}}:\,\,\Rank(\overline{\mathbf{W}}_\mathrm{I})\leq1$. Besides, constraint $\overline{\mathrm{C6}}$ involves infinitely many constraints due to the continuous uncertainty set $\bm{\Psi}_j$. Next, we introduce a Lemma which  allows us to transform constraint $\overline{\mathrm{C6}}$ into a finite number of linear matrix inequalities
(LMIs) constraints.
\begin{Lem}[S-Procedure \cite{book:convex}] Let a function $f_m(\mathbf{x}),m\in\{1,2\},\mathbf{x}\in \mathbb{C}^{N\times 1},$ be defined as
\begin{eqnarray}
f_m(\mathbf{x})=\mathbf{x}^H\mathbf{A}_m\mathbf{x}+2 \mathrm{Re} \{\mathbf{b}_m^H\mathbf{x}\}+c_m,
\end{eqnarray}
where $\mathbf{A}_m\in\mathbb{H}^N$, $\mathbf{b}_m\in\mathbb{C}^{N\times 1}$, and $c_m\in\mathbb{R}$. Then, the implication $f_1(\mathbf{x})\le 0\Rightarrow f_2(\mathbf{x})\le 0$  holds if and only if there exists an $\omega\ge 0$ such that
\begin{eqnarray}\omega
\begin{bmatrix}
       \mathbf{A}_1 & \mathbf{b}_1          \\
       \mathbf{b}_1^H & c_1           \\
           \end{bmatrix} -\begin{bmatrix}
       \mathbf{A}_2 & \mathbf{b}_2          \\
       \mathbf{b}_2^H & c_2           \\
           \end{bmatrix}          \succeq \mathbf{0},
\end{eqnarray}
provided that there exists a point $\mathbf{\hat{x}}$ such that $f_k(\mathbf{\hat{x}})<0$.
\end{Lem}
Now, we apply Lemma 1 to constraint $\overline{\mathrm{C6}}$. In particular, we define $\mathbf{\widehat g}_j=\vect(\mathbf{\widehat G}_j)$, $\Delta\mathbf{ g}_j=\vect(\Delta\mathbf{G}_j)$,  $\widetilde{{\mathbf{W}}}_{\mathrm{I}}=\mathbf{I}_{N_\mathrm{R}}\otimes\overline{\mathbf{W}}_{\mathrm{I}}$, and $\widetilde{{\mathbf{W}}}_{\mathrm{E}}=\mathbf{I}_{N_\mathrm{R}}\otimes\mathbf{W}_{\mathrm{E}}$.  By exploiting  the fact that $\norm{\Delta\mathbf{G}_j}_F^2 \le \varepsilon_j^2 \Leftrightarrow \Delta\mathbf{g}_j^H \Delta\mathbf{g}_j\le \varepsilon_j^2$, then we have
\begin{eqnarray}
\hspace*{-4.8mm}&&\hspace*{-0.8mm}\norm{\Delta\mathbf{G}_j}_F^2 \le \varepsilon_j^2\\
\Rightarrow\,\hspace*{-4.8mm}&&\hspace*{-0.8mm}\overline{\mathrm{C6}}:  0 \ge   \gamma_j+ \min_{\Delta\mathbf{g}_j\in {\bm\Psi}_j} -\Big\{\Delta\mathbf{g}_j^H\Big(\widetilde{{\mathbf{W}}}_{\mathrm{I}}+\widetilde{{\mathbf{W}}}_{\mathrm{E}}\Big)\Delta\mathbf{g}_j\notag\\
&&\hspace*{-5mm}+
2\mathrm{Re}\Big\{\mathbf{\widehat g}_j^H\Big(\widetilde{{\mathbf{W}}}_{\mathrm{I}}+\widetilde{{\mathbf{W}}}_{\mathrm{E}}\Big)\Delta\mathbf{g}_j\Big\}+\mathbf{\widehat g}_j^H\Big(\widetilde{{\mathbf{W}}}_{\mathrm{I}}+\widetilde{{\mathbf{W}}}_{\mathrm{E}}\Big)\mathbf{\widehat g}_j \Big\},\forall j,\notag
\end{eqnarray}
if and only if there exists a $\rho_j\ge 0$ such that the following
LMIs constraint holds:
\begin{eqnarray}
\label{eqn:LMI_6}
\mbox{\textoverline{C6}: }\mathbf{S}_{\mathrm{\overline{C6}}_j}
\hspace*{-2.5mm}&=&\hspace*{-2.5mm}
         \begin{bmatrix}
       \rho_j\mathbf{I}_{N_{\mathrm{T}}}\hspace*{-0.5mm}+\hspace*{-0.5mm}\widetilde{{\mathbf{W}}}_{\mathrm{E}}& \hspace*{-0.5mm}\widetilde{{\mathbf{W}}}_{\mathrm{E}}\mathbf{\hat g}_j          \\ \notag
       \mathbf{\hat g}_j^H \widetilde{{\mathbf{W}}}_{\mathrm{E}}
        & \hspace*{-0.5mm}-\hspace*{-0.5mm}\rho_j\varepsilon_j^2 \hspace*{-0.5mm}-\hspace*{-0.5mm}\frac{\gamma_j}{\eta_j}\hspace*{-0.5mm} +\hspace*{-0.5mm} \mathbf{\hat g}_j^H \widetilde{{\mathbf{W}}}_{\mathrm{E}} \mathbf{\hat g}_j        \\
           \end{bmatrix}\\
           \hspace*{-0.5mm}&+&\hspace*{-0.5mm} \mathbf{U}_{\mathbf{g}_j}^H\widetilde{{\mathbf{W}}}_{\mathrm{I}}\mathbf{U}_{\mathbf{g}_j}\hspace*{-0.5mm}\succeq \hspace*{-0.5mm}\mathbf{0}, \forall k,
          \end{eqnarray}
where $\mathbf{U}_{\mathbf{g}_j}=\big[\mathbf{I}_{N_{\mathrm{R}}N_{\mathrm{T}}},\,\, \mathbf{\widehat g}_j\big]$. The new  constraint $\mbox{\textoverline{C6}}$ is not only an affine function  with respect to the optimization variables, but also involves only a finite number of constraints. Then, we apply the SDP relaxation by removing constraint $\mbox{\textoverline{C3}}$ from Problem \ref{prob:multiobj_WIPTsepuser2}. As a result, the SDP relaxed problem is given by
\begin{framed}
\begin{Prob}SDP Relaxed Transformed MOOP:\label{prob:multiobj_WIPTsepuser_relaxed}
\begin{eqnarray}
\underset{\overline{\mathbf{W}}_\mathrm{I},\overline{\mathbf{W}}_\mathrm{E}\in\mathbb{H}^{N_\mathrm{T}},
\theta,\tau,\gamma_j,\rho_j}{\mino}&&\tau \\
\mathrm{subject\,\,to}\,\,\, &&\overline{\mathrm{C1}},\,\overline{\mathrm{C2}},\,\overline{\mathrm{C4}},\,\overline{\mathrm{C5}}, \overline{\mathrm{C7}},\,\notag\\
&&\overline{\mathrm{C6}}: \mathbf{S}_{\mathrm{\overline{C6}}_j}\succeq\zero, \overline{\mathrm{C8}}: \rho_j\geq 0,\forall j,\notag
\end{eqnarray}
\end{Prob}\vspace*{-0.5mm}\end{framed}
\noindent which is a convex SDP problem and can be solved by numerical convex program solvers such as CVX \cite{website:CVX}. In particular, if the obtained solution $\overline{\mathbf{W}}_\mathrm{I}^*$ of the SDP relaxed problem satisfies constraint $\overline{\mathrm{C3}}$, i.e., $\Rank(\overline{\mathbf{W}}_\mathrm{I}^*)\leq1$, then it is the optimal solution. Now, we study the tightness of the SDP relaxation by the following theorem.
\begin{Thm}\label{thm:rankone}
Assuming that the channels, i.e., $\mathbf{h}$ and $\mathbf{G}_j$,  are statistically independent and Problem \ref{prob:multiobj_WIPTsepuser_relaxed} is feasible,
the optimal beamforming matrix of Problem \ref{prob:multiobj_WIPTsepuser_relaxed} is a rank-one matrix with probability one, i.e.,  $\Rank(\overline{\mathbf{W}}_\mathrm{I}^*)\le1$. Besides, for $\omega_1>0$, the optimal energy signal is $\overline{\mathbf{w}}_\mathrm{E}^*=\zero$.
\end{Thm}
\begin{proof}
Please refer to the Appendix.
\end{proof}
Therefore, the adopted SDP relaxation is tight. Besides, whenever AR-EE is considered, i.e., $\omega_1>0$, no dedicated energy beam is needed. In fact, the optimal information beam, $\overline{\mathbf{w}}_\mathrm{I}$, serves as a dual purpose carrier for maximization of the energy efficiency of information transmission and  WPT simultaneously. Furthermore, Problems 1-2 can be solved by SDP relaxation as solving Problem \ref{prob:multiobj_WIPTsepuser_relaxed}.

\section{Results}
In this section, we present simulation results to demonstrate the system performance of multi-objective SWIPT system design. The simulation parameters are summarized in Table \ref{table:parameters}.  The IR and $J$ ERs are located $100$ meters and $10$ meters from the transmitter.  In particular, the ERs are near the transmitter with line-of-sight communication channels to facilitate energy harvesting. Each ER is equipped with $N_{\mathrm{R}}=2$ antennas for facilitating EH. We assume that the noise powers at each antenna of the IR and the ERs are identical, i.e., $\sigma_\mathrm{I}^2=\sigma_\mathrm{E}^2=\sigma^2$.  In the sequel, we define the normalized maximum  channel estimation error of ER $j$  as  $\delta_{j}^2=\frac{\varepsilon^2_j}{\norm{\mathbf{G}_j}^2_F}$  with $\delta_{a}^2=\delta_{b}^2=0.05,\forall a, b\in\{1,\ldots,J\}$.    All simulation results are obtained by averaging the system performance over different multipath channel realizations.
\begin{table}[t]
\caption{Simulation Parameters} \label{table:parameters}
\centering
\begin{tabular}{ | l | l | } \hline
      Carrier center frequency                           & 915 MHz\\ \hline
      Bandwidth                                          & $200$ kHz \\ \hline 
      Single antenna power consumption                   & $P_\mathrm{ant}=1$ W \\ \hline
      Static circuit power consumption                           & $P_\mathrm{c}=150$ W \cite{Jr:power_consumption} \\ \hline
      Power amplifier efficiency                         & $\xi=0.2$ \\ \hline
      Transmit  antenna gain                                     & 18 dBi \\ \hline
      Noise power                                        & $\sigma^2= -95$ dBm \\ \hline
      Transmitter-to-ERs fading distribution                                      & Rician with Rician factor $6$ dB \\
        \hline
       Transmitter-to-IR fading distribution                                      & Rayleigh \\
      \hline
      Energy conversion efficiency                       & $\eta_j=0.5$ \\ \hline
\end{tabular}
\end{table}

\begin{figure}[t]
        \centering
       \includegraphics[width=3.5 in]{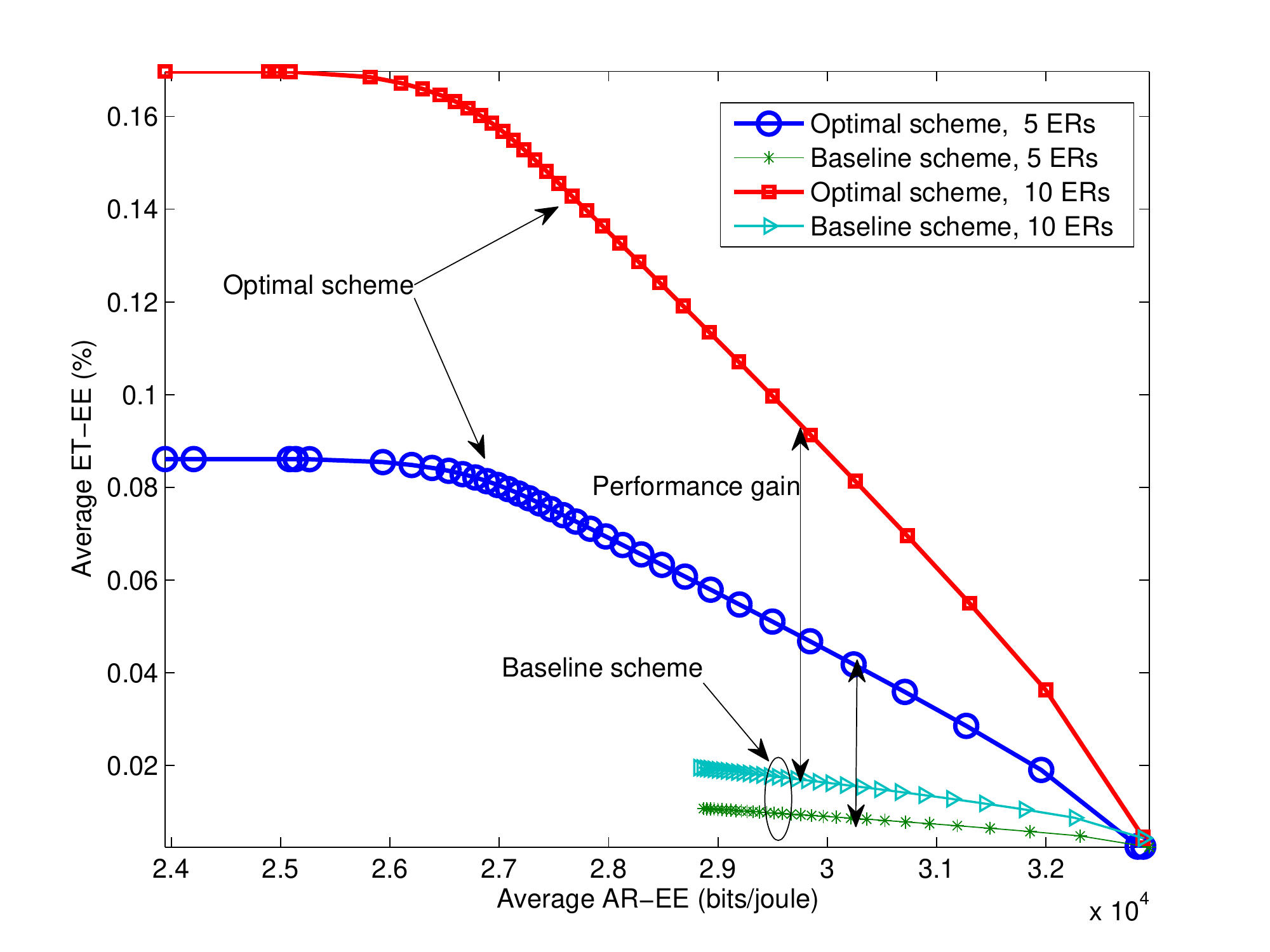}
        \caption{System performance trade-off region between the average  ET-EE  and the average AR-EE for $P_{\max}=40$ dBm. The double-sided arrows indicate the performance gain achieved by the proposed optimal scheme over the baseline scheme.}
        \label{fig:trade-off}
\end{figure}

Figure \ref{fig:trade-off} depicts the  trade-off region  for
 the average ET-EE and the average AR-EE  achieved by the proposed optimal scheme for different numbers of ERs. The maximum transmit power is set to $P_{\max}=40$ dBm.
The trade-off region in  Figure \ref{fig:trade-off} is obtained by solving Problem $7$
 via varying the values of $0\le \lambda_p\le 1,\forall p\in\{1,2\}$,   uniformly  for a step size of $0.01$ such that $\sum_p \lambda_p=1$.
It can be observed that the ET-EE is a monotonically decreasing function
 with respect to the  AR-EE. The result indicates that AR-EE maximization and ET-EE maximization are conflicting system design objectives in general.  In other words, a resource allocation algorithm maximizing the AR-EE cannot maximize the ET-EE simultaneously in the considered system.  Besides, the trade-off region is enlarged for an increasing number of ERs.  This is due to the fact that
a larger portion of the radiated power can be harvested when there are
more ERs in the system since more receivers  participate
in the energy harvesting process.

For comparison, we also plot the trade-off region of a baseline
power allocation scheme  in Figure
\ref{fig:trade-off}. For the baseline scheme, the covariance matrix of the energy signal  $\overline{\mathbf{W}}_{\mathrm{E}}$ is set to zero. Then,  maximum ratio transmission (MRT) with respect to the IR is adopted for the information beamforming matrix $\overline{\mathbf{W}}_{\mathrm{I}}$. In other words, the beamforming direction of matrix $\overline{\mathbf{W}}_\mathrm{I}$ is fixed. Then, we optimize  the power of  $\overline{\mathbf{W}}_\mathrm{I}$ subject to the constraints in Problem $7$.  It can be observed
that the  baseline scheme achieves a significantly smaller trade-off region compared to the proposed optimal scheme. As a matter of fact, the degrees of freedom of the beamforming matrix $\overline{\mathbf{W}}_\mathrm{I}$ are jointly optimized in our proposed optimal scheme via utilizing the CSI of all receivers.  On the contrary,  the information beamformer in the baseline scheme  is restricted to the range space of the IR. Although the baseline scheme is optimal when AR-EE is the only system design objective, the information beamformer cannot be steered towards the direction of the ERs. Thus, compared to the proposed optimal scheme, the baseline scheme is less efficient when ET-EE is considered.

\begin{figure}[t]\vspace*{-3mm}
        \includegraphics[width=3.5 in]{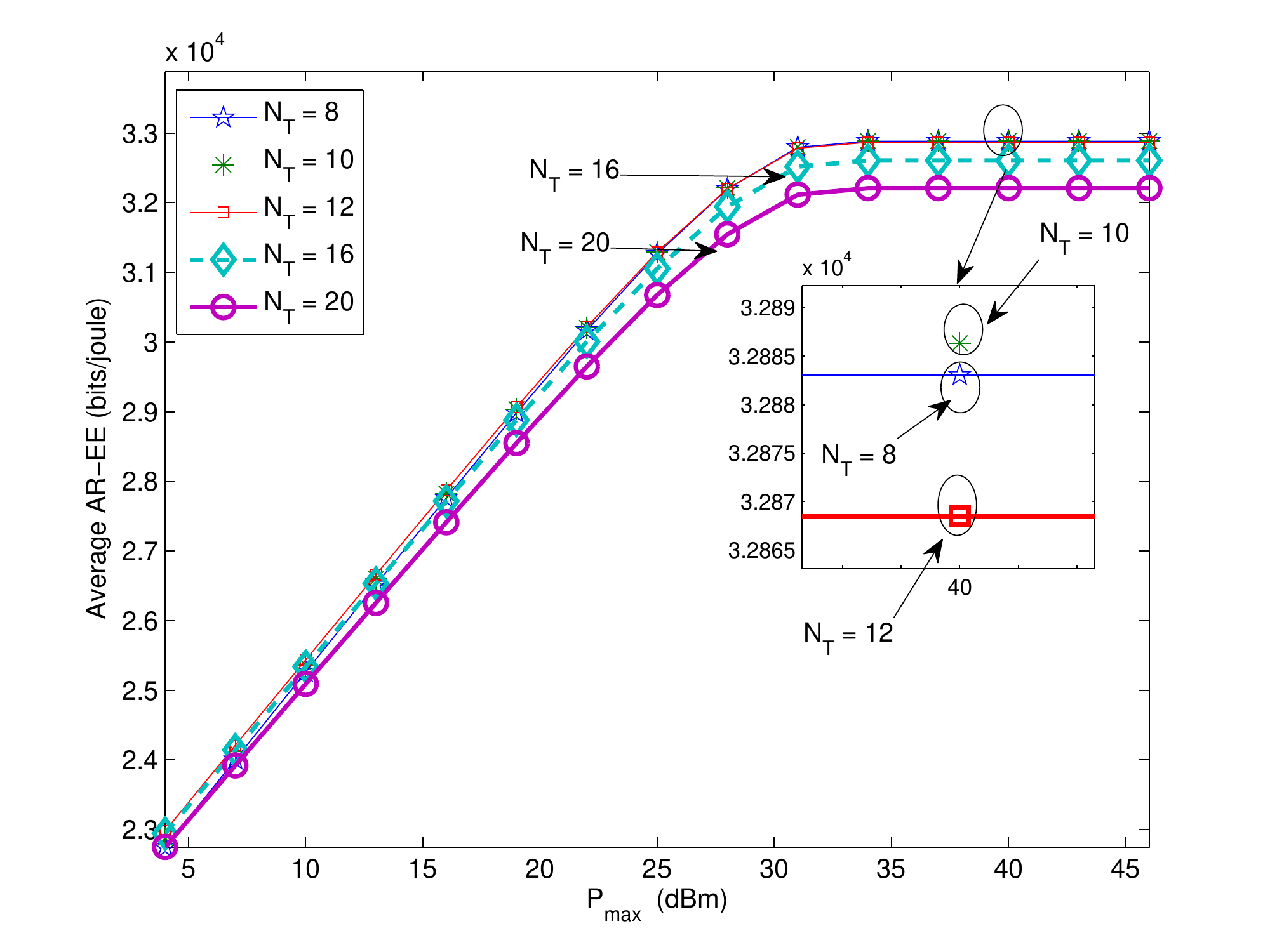}
        \caption{Average AR-EE (bits/joule) versus maximum transmit power budget $P_{\max}$ (dBm).}
        \label{fig:trade-off1}\vspace*{-3mm}
\end{figure}

Figures \ref{fig:trade-off1} and \ref{fig:trade-off2} illustrate  the average AR-EE and the average ET-EE versus the total transmit power budget $P_{\max}$, respectively, for $J=10$ ERs. The results in Figures \ref{fig:trade-off1} and \ref{fig:trade-off2} are obtained by solving Problem $7$ with $\{\omega_1=1,\omega_2=0\}$ and  $\{\omega_1=0,\omega_2=1\}$, respectively. It can be observed from Figure \ref{fig:trade-off1} that the AR-EE of the proposed optimal scheme increases with respect to $P_{\max}$ monotonically  and  reaches an upper limit where the EE gain due to a higher value of $P_{\max}$ vanishes. This result indicates that once the maximum AR-EE is achieved by transmitting a sufficiently large power, any additional increase in the transmitted power will incur a loss in EE which is prevented by the optimal algorithm. On the other hand, it can be seen from Figure \ref{fig:trade-off2} that the average ET-EE increases slowly in the low transmit power regime but increases rapidly in the high transmit power regime. This is because for a small transmit power, the ET-EE is dominated by the fixed circuit power consumption, $P_{\mathrm{B}}$, leading to a slow increasing rate of ET-EE with respect to the transmit power. As the transmit power budget increases, the transmit power consumption in the RF becomes significant and the ET-EE becomes more sensitive to increases in transmit power budget, cf. \eqref{eq:ET-EE}. On the other hand, the number of transmit antennas $N_{\mathrm{T}}$ affects the AR-EE and the ET-EE differently. In fact, the maximum AR-EE does not necessarily increase with the number of transmit antennas when the per-antenna power consumption in considered, cf. Figure \ref{fig:trade-off1}.  This is because the AR scales logarithmically with respect to the number of transmit antennas. However, the AR gain  due to extra transmit antennas is not sufficient to compensate the total increased energy cost since the circuit power consumption increases linearly with respect to  $N_{\mathrm{T}}$. Thus,  adopting exceedingly large numbers of transmit antennas may not be a viable solution for information transmission. In contrast, the maximum  ET-EE increases with $N_{\mathrm{T}}$ as shown in Figure \ref{fig:trade-off2}. This is due to the fact that the ET-EE function in \eqref{eq:ET-EE} is a quasi-linear function with respect to both $\overline{\mathbf{W}}_\mathrm{I}$ and $\overline{\mathbf{W}}_\mathrm{E}$. Thus, a large number of transmit antennas is beneficial if ET-EE is the only system design objective.

\begin{figure}[t]\vspace*{-3mm}
        \includegraphics[width=3.5 in]{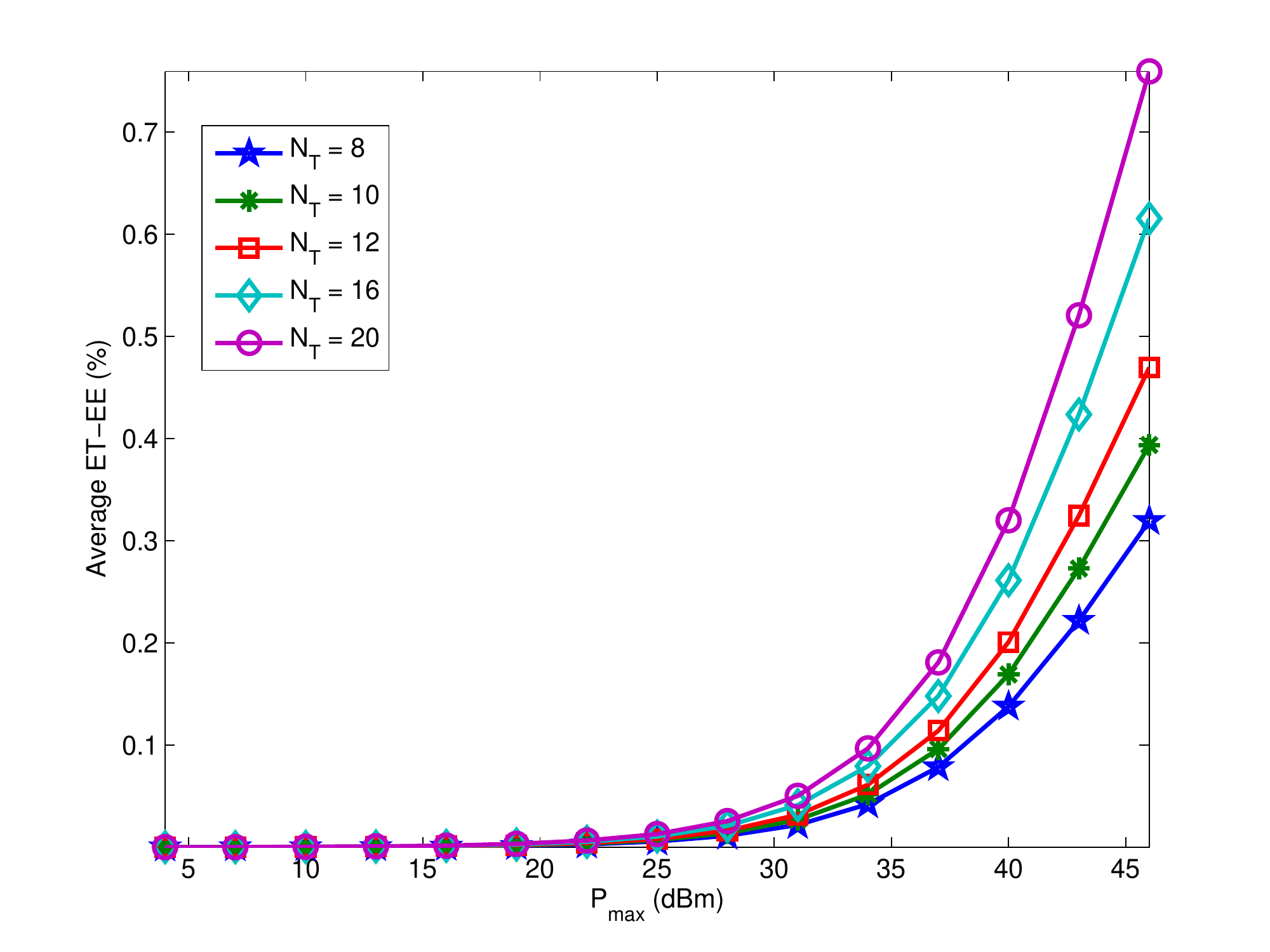}
        \caption{Average ET-EE  versus maximum transmit power budget $P_{\max}$ (dBm).}
        \label{fig:trade-off2}\vspace*{-3mm}
\end{figure}
\section{Conclusions}\label{sect:conclusion}
In this paper, we studied the resource allocation algorithm design for energy-efficient SWIPT networks.
 The algorithm design was formulated as a non-convex  MOOP employing the weighted Tchebycheff method.  The proposed problem aimed at the joint maximization of the energy efficiency of information transmission and  WPT simultaneously. Besides, the imperfectness of the CSI of the ERs was also taken into account for designing a robust resource allocation algorithm. The proposed MOOP was solved optimally by SDP relaxation. Simulation results not only reveal the trade-off between the studied conflicting systems design objectives, but also shed some light on the use of multiple transmit antennas for improving system energy efficiency.

\section*{Appendix-Proof of Theorem \ref{thm:rankone}}\label{app:rankone}
 Since Problem \ref{prob:multiobj_WIPTsepuser_relaxed} satisfies Slater's constraint qualification and is jointly convex with respect to the optimization variables, strong duality holds. Thus, solving the dual problem is equivalent to solving the primal problem. Therefore,  Theorem \ref{thm:rankone} can be proved by analyzing the dual problem of Problem \ref{prob:multiobj_WIPTsepuser_relaxed}. To this end, we define the Lagrangian function ${\cal L}$
\begin{eqnarray} \label{eqn:appB7}
&=&\tau+\alpha\big(\Tr(\overline{\mathbf{W}}_\mathrm{I}+\overline{\mathbf{W}}_\mathrm{E})-\theta P_{\mathrm{max}}\big)-\delta\theta\notag\\
&+&\beta\big(\frac{\Tr(\overline{\mathbf{W}}_\mathrm{I}+\overline{\mathbf{W}}_\mathrm{E})}{\xi}+
\theta P_{\mathrm{B}}-1\big)\notag\\
&+&\upsilon_1\Big[\omega_1\big(F_1^*-\theta\log_2(1+\frac{\Tr(\mathbf{H}\overline{\mathbf{W}}_\mathrm{I})}{\theta\sigma_\mathrm{I}^2})\big)-\tau\Big]\notag\\
&+&\upsilon_2\Big[\omega_2\big(F_2^*-\sum_{j=1}^J\gamma_j\big)-\tau\Big]-\Tr(\mathbf{X}\overline{\mathbf{W}}_\mathrm{I})-\Tr(\mathbf{Y}\overline{\mathbf{W}}_\mathrm{E})\notag\\
&-&\sum_{j=1}^J\Tr(\mathbf{D}_{\mathrm{\overline{C6}}_j}\mathbf{S}_{\mathrm{\overline{C6}}_j})-\sum_{j=1}^J\mu_j\rho_j,
\end{eqnarray}
where $\alpha,\beta,\delta,\mu_j\ge 0$ are dual variables associated with  constraints $\overline{\mathrm{C1}}$, $\overline{\mathrm{C4}}$, $\overline{\mathrm{C5}}$, and $\overline{\mathrm{C8}}$, respectively. Dual variable matrices $\mathbf{X},\mathbf{Y}$, and $\mathbf{D}_{\mathrm{\overline{C6}}_j}$ are connected to the LMI constraints in $\overline{\mathrm{C2}}$ and $\overline{\mathrm{C6}}$, respectively. $\upsilon_1,\upsilon_2$ are the dual variables for constraint $\overline{\mathrm{C7}}$.

Then, the dual problem of  Problem \ref{prob:multiobj_WIPTsepuser_relaxed} is given by
\begin{eqnarray}\label{eqn:dual}
\hspace*{-1cm}\underset{ {\alpha,\beta,\delta,\mu_j\ge 0,\mathbf{Y},\mathbf{X},\mathbf{D}_{\mathrm{\overline{C6}}_j}\succeq \mathbf{0}}}{\maxo} \quad\underset{\underset{\theta,\tau,\gamma_j,\rho_j}{\overline{\mathbf{W}}_\mathrm{I},\overline{\mathbf{W}}_\mathrm{E}
\in\mathbb{H}^{N_{\mathrm{T}}},}}{\mino}\,\,\,{\cal
L}.\label{eqn:master_problem}
\end{eqnarray}

  Now, we focus on those Karush-Kuhn-Tucker (KKT) conditions which are useful for the proof:
\begin{eqnarray}
\hspace*{-3mm}\mathbf{Y},\mathbf{X},\mathbf{D}_{\mathrm{\overline{C6}}_j}\hspace*{-3mm}&\succeq&\hspace*{-3mm}\mathbf{0},\,\,\alpha,\beta,\delta,\mu_j\ge 0,\,\forall j, \label{eqn:dual_variables}\\
\mathbf{X}\hspace*{-3mm}&=&\hspace*{-3mm}\mathbf{Y}-\frac{\upsilon_1\omega_1\theta}{\theta\sigma_\mathrm{I}^2+\Tr(\mathbf{H}\overline{\mathbf{W}}_\mathrm{I})}\mathbf{H},\label{eqn:appB1}\\
\mathbf{Y}\hspace*{-3mm}&=&\hspace*{-3mm}(\hspace*{-0.5mm}\alpha\hspace*{-0.5mm}+\hspace*{-0.5mm}\frac{\beta}{\xi})\mathbf{I}_{N_{\mathrm{T}}}
\hspace*{-1.5mm}-\hspace*{-1.5mm}\sum_{j=1}^{J}\hspace*{-0.5mm}\sum_{l=1}^{N_{\mathrm{R}}}\hspace*{-0.5mm}\Big[\hspace*{-0.5mm}
\mathbf{U}_{\mathbf{g}_j}\mathbf{D}_{\mathrm{\overline{C6}}_j}\hspace*{-0.5mm}\mathbf{U}_{\mathbf{g}_j}^H \hspace*{-0.5mm}\Big]_{a:b,c:d},\label{eqn:appB2}\\
\mathbf{X}\overline{\mathbf{W}}_\mathrm{I}\hspace*{-3mm}&=&\hspace*{-3mm}\mathbf{0},\label{eqn:appB3}\\
\mathbf{Y}\overline{\mathbf{W}}_\mathrm{E}\hspace*{-3mm}&=&\hspace*{-3mm}\mathbf{0}.\label{eqn:appB4}
\end{eqnarray}
where $a=(l-1)N_{\mathrm{T}}+1,b=l N_{\mathrm{T}},c=(l-1)N_{\mathrm{T}}+1,$ and $d=l N_{\mathrm{T}}$. Next, we investigate the structure of  $\overline{\mathbf{W}}_\mathrm{I}$  in the following two cases.

\textbf{Case 1}: For $\omega_1=0$ and $\omega_2=1$,  AR-EE maximization is not considered and ET-EE maximization is the only system design objective. Besides, since $\omega_1=0$, we have $\mathbf{X}=\mathbf{Y}$. In other words, $\overline{\mathbf{W}}_\mathrm{I}$ has the same functionality as $\overline{\mathbf{W}}_\mathrm{E}$ for improving the ET-EE. Thus, without loss of generality and optimality, we can  set $\overline{\mathbf{W}}_\mathrm{I}=\zero$ and $\Rank(\overline{\mathbf{W}}_\mathrm{I})\leq 1$ holds for the optimal solution.

\textbf{Case 2}: For $\omega_1>0$, AR-EE maximization is considered in the resource allocation algorithm design.
Thus, constraint $\overline{\mathrm{C6}}$ for $j=1$ is active, i.e., $\upsilon_1>0$. Besides, from the complementary slackness condition in \eqref{eqn:appB3}, the columns of $\overline{\mathbf{W}}_\mathrm{I}$ for the optimal solution lie in the null space of $\mathbf{X}$ for $\overline{\mathbf{W}}_\mathrm{I}\ne\mathbf{0}$. Therefore, if  $\Rank(\mathbf{X})=N_{\mathrm{T}}-1$, then the optimal beamforming matrix must be a rank-one matrix. To reveal the structure of $\mathbf{X}$, we show by contradiction that $\mathbf{Y}$ is a positive definite matrix with probability one. For a given set of optimal dual variables, $\{\mathbf{Y},\mathbf{X},\mathbf{D}_{\mathrm{\overline{C6}}_j},\alpha,\beta,\delta,\mu_j \}$, and optimal primal variables $\{\tau,\gamma_j,\theta,\overline{\mathbf{W}}_\mathrm{E},\rho_j\}$,    (\ref{eqn:dual}) can be written as
\begin{eqnarray}\hspace*{-2mm}\label{eqn:dual2}
\,\,\underset{\overline{\mathbf{W}}_\mathrm{I}\in\mathbb{H}^{N_{\mathrm{T}}}}{\mino} \,\, {\cal L}.
\end{eqnarray}
Suppose $\mathbf{Y}$ is not positive definite, then we can construct $\overline{\mathbf{W}}_\mathrm{I}=r\mathbf{v}\mathbf{v}^H$ as one of the optimal solutions of (\ref{eqn:dual2}), where $r>0$ is a scaling parameter and $\mathbf{v}$ is the eigenvector corresponding to one of the non-positive eigenvalues of $\mathbf{Y}$.  We substitute $\overline{\mathbf{W}}_\mathrm{I}=r\mathbf{v}\mathbf{v}^H$ into (\ref{eqn:dual2}) which leads to
${\cal L}=\Tr(r\mathbf{Y}\mathbf{v}\mathbf{v}^H)-r\Tr\Big(\mathbf{v}\mathbf{v}^H
\big(\mathbf{X}+\frac{\upsilon_1\omega_1\theta\mathbf{H}}{\theta\sigma_\mathrm{I}^2+\Tr(\mathbf{H}\overline{\mathbf{W}}_\mathrm{I})}\big)\Big)+\Omega$ where $\Omega$ is the collection of variables that are independent of $\overline{\mathbf{W}}_\mathrm{I}$. Since the channels of $\mathbf{G}_j$ and $\mathbf{h}$ are assumed to be statistically independent, it follows that by setting $r\rightarrow \infty$, the dual optimal value  becomes unbounded from below. Besides, the optimal value of the primal problem is non-negative. Thus,  this leads to a contradiction since strong duality does not hold. Therefore, $\mathbf{Y}$ is a positive definite matrix with probability one, i.e., $\Rank(\mathbf{Y})=N_{\mathrm{T}}$.

By exploiting (\ref{eqn:appB1}) and a basic inequality for the rank of matrices, we have
\begin{eqnarray}
\hspace*{-3mm}&&\hspace*{-2mm}\Rank(\mathbf{X})+\Rank\Big(\frac{\upsilon_1\omega_1\theta\mathbf{H}}{\theta\sigma_\mathrm{I}^2+\Tr(\mathbf{H}\overline{\mathbf{W}}_\mathrm{I})}\Big)\\
&\ge&\hspace*{-2mm} \Rank\big(\mathbf{Y}\big)=N_\mathrm{T}\\
\hspace*{-3mm}&\Rightarrow &\hspace*{-2mm}
 \Rank(\mathbf{X})\ge N_{\mathrm{T}}-\Rank\Big(\frac{\upsilon_1\omega_1\theta\mathbf{H}}{\theta\sigma_\mathrm{I}^2+
 \Tr(\mathbf{H}\overline{\mathbf{W}}_\mathrm{I})}\Big).
\end{eqnarray}
Furthermore, $\overline{\mathbf{W}}_\mathrm{I}\ne\mathbf{0}$ is required to maximize the energy efficiency of data communication. Hence, $\Rank(\mathbf{X})=N_{\mathrm{T}}-1$ and $\Rank(\overline{\mathbf{W}}_\mathrm{I})=1$. Besides, in this case, since $\mathbf{Y}$ is full rank, $\overline{\mathbf{W}}_\mathrm{E}=\zero$ according to \eqref{eqn:appB4}. In other words, utilizing only information beam $\overline{\mathbf{w}}_\mathrm{I}$ is optimal when $\omega_1>0$.

\bibliographystyle{IEEEtran}
\bibliography{vicky_leng}

\end{document}